\documentclass[onecollarge]{svjour2}       

\smartqed  

\usepackage{color}
\usepackage{graphicx}
\usepackage{multirow}
\usepackage{amssymb}
\usepackage{amsmath}
\usepackage{epsfig}
\newcommand\numbereq{\addtocounter{equation}{1}\tag{\theequation}}

\newcommand{\specialcell}[2][c]{\begin{tabular}[#1]{@{}c@{}}#2\end{tabular}}

\newtheorem{prop}{Proposition}
	
\begin{document}

\title{Modified cumulative distribution function in application to waiting time analysis in CTRW scenario
}


\author{Rafa{\l} Po{\l}ocza{\'n}ski \and Agnieszka Wy{\l}oma{\'n}ska \and Janusz Gajda \and Monika Maciejewska \and Andrzej Szczurek
}


\institute{R. Po{\l}ocza{\'n}ski, A. Wy{\l}oma{\'n}ska, J. Gajda \at
              Faculty of Pure and Applied Mathematics, Hugo Steinhaus Center, Wroc{\l}aw University of Technology \\
              Wybrze{\.z}e Wyspia{\'n}skiego, 27, 50-370 Wroc{\l}aw, Poland\\
                            \email{rafal.poloczanski,agnieszka.wylomanska,janusz.gajda@pwr.edu.pl}           
           \and
           M. Maciejewska, A. Szczurek \at
					Institute of Air Conditioning and District Heating, Wroc{\l}aw University of Technology\\
              Wybrze{\.z}e Wyspia{\'n}skiego, 27, 50-370 Wroc{\l}aw, Poland\\
                            \email{monika.maciejewska,andrzej.szczurek@pwr.edu.pl}  
}

\date{Received: date / Accepted: date}

\maketitle

\begin{abstract}
The continuous time random walk model plays an important role in modeling of so called anomalous
diffusion behaviour. One of the specific property of such model are constant time periods visible in trajectory. In the continuous time random walk approach they are realizations of the sequence called waiting times. The main attention of the paper is paid on the analysis of waiting times distribution. We introduce here  novel methods of estimation and statistical investigation of such distribution. The methods are based on the modified cumulative distribution function. In this paper we consider three special cases of waiting time distributions, namely $\alpha$-stable, tempered stable and gamma. However the proposed methodology can be applied to broad set of distributions - in general it may serve as a method of fitting any distribution function if the observations are rounded.  The new statistical techniques we apply to the simulated data as well as to the real data describing $CO_2$ concentration in indoor air.  
\end{abstract}

\section{Introduction}
\label{sec:introduction}

In the real data analysis the selection of appropriate model suitable to examined  vector of observations is the most important issue. The proper model should take under consideration the important properties of examined data. The one of the simplest models is based on the assumption that the vector of observations constitute a sample of independent identically distributed random variables.  However, for most time series the theoretical model is much more complicated.  Therefore there is a need to analyse more sophisticated systems in order to cover all properties of the examined signal. 

For some analysed data very often we observe specific behaviour, namely there are visible constant time periods. This special property may indicate that the theoretical model behind the time series is based on the so called continuous time random walk. This model in the classical version is constructed as a sum of independent random variables (called later jumps) from the same distribution however the number of jumps in the mentioned sum is a realization of the process based on the so called waiting time sequence. Therefore the continuous time random walk model is defined through the jumps as well as waiting time distributions.

In the continuous time random walk scheme, originally introduced in \cite{bib:Montrol1965}, the waiting times between the jumps are not constant, as in the standard random walk, but are random variables governed by some
probability law. As a consequence, the continuous time random walk model is a natural description of
transport in crowded environments and complex systems. It is worth to mention, the limit distribution of the continuous time random walk scheme can be also formulated within the framework of the fractional Fokker-Planck equation \cite{metzler2000}. At the same time, the limit of the scaled continuous time random walk  is a process (called subordinated process) which is a superposition of two systems: one described by Langevin type equations and the second one, called inverse subordinator \cite{sub1,sub2}. In the recent years the subordinated processes more and more often appear in the literature \cite{sub3,sub4,sub5}. 

The continuous time random walk model as well as its limit (subordinated processes) is one of the most important system that exhibits anomalous diffusion property.  Anomalous diffusion models
have found many practical applications. They were used in variety of
physical systems, including charge carrier transport in amorphous semiconductors
\cite{a2,a3,a4}, transport in micelles \cite{a5}, intracellular transport \cite{a6} or motion
of mRNA molecules inside E. coli cells \cite{a7}. The behaviour corresponding to continuous time random walk scenario  can be also observed in  stock prices or interest rates data \cite{bib:Orzel2011,Janczura_orzel} as well as in the time series related to indoor air quality parameters \cite{bib:Maciejewska2012,bib:Janczura2013,bib:Maciejewska2015}.

One of the most important issues that arises in the analysis of the continuous time random walk model is the description of waiting times that correspond to
the periods of constant values visible in the data. Finding a proper waiting times distribution
allows to conclude on the properties of the whole process. The most popular
 distribution of the constant time periods is the  $\alpha$-stable \cite{Janczura_orzel,sub3},
but recent developments  indicate that another non-negative infinitely
divisible distributions can be also used to model the observed waiting times,
\cite{bib:MagdziarzGajda2010,inne2,inne3,bib:Janczura2011,sub1}. 

There are known successful attempts of applying waiting time analysis in many areas of research and practice.  In this work we demonstrate its usefulness in the domain of measurements. 
A continuous monitoring of various quantities is used in many applications as preferred to other type of measurements \cite{bib:monitoring1,bib:monitoring2}. Usually, it provides a great number of discrete data. They have to be transmitted, verified and analysed \cite{bib:performance}. These operations consume significant amounts of time and energy. Therefore, they pose a serious problem in many monitoring systems, especially in wireless networks \cite{bib:wireless}. Additionally, continuous monitoring may be a source of useless data. It results from the measurement characteristics of sensors e.g. inappropriate accuracy and sensitivity \cite{bib:Szczurek2015}. These disadvantages can be reduced by the application of a suitable sampling procedure \cite{bib:sampling}. Sampling indicates how much data to collect and how often it should be collected in order to gain access to the requested information. In this work, it is demonstrated that the statistical analysis of waiting time for an essential change of $CO_2$ concentration inside room may be applied to specify the most appropriate time interval between consecutive measurements. In this way, the sampling rate (frequency of sampling) is determined. It was assumed that the sampling procedure could reduce the amount of data. However, the information content cannot be significantly reduced. 

The main goal of the paper is to introduce a new method of estimation and statistical investigation of waiting times distribution in the continuous time random walk scenario.  The  most commonly used method to identify the correct distribution of constant time periods is visualization of empirical cumulative distribution function  and comparison with fitted parametric distribution \cite{bib:Wylomanska2012}. However,  testing procedures which would allow to identify the correct class of distribution for observations corresponding to waiting times  are very rarely considered in the literature. The difficulty arises from the discretization of the time intervals which is a natural consequence of time-discretization of measuring devices. Moreover, standard estimation procedures  might lead to under- or overestimation of parameters. The proposed methods of estimation and statistical testing of proper distribution of waiting times in continuous time random walk scenario  are based on the modified cumulative distribution function. Here we propose a novel techniques that can be applied to broad set of distributions - in general it may serve as a method of fitting any distribution function if the observations are rounded. In our analysis we consider three waiting times distributions, namely $\alpha$-stable, tempered stable and gamma, as the most commonly used to constant time periods description.

The rest of the paper is organized as follows: in Section \ref{sec:model} we introduce the continuous time random walk model with three possible waiting time distributions, namely $\alpha-$stable, tempered stable and gamma and indicate the main properties of the considered models. Next, in Section \ref{sec:estimation} we introduce the modified cumulative distribution function which is a base for estimation and testing of proper distribution for constant time periods. In Section \ref{sec:simulation} we present how the introduced methods work for simulated data from continuous time random walk while in Section \ref{sec:analysis} we analyse the real time series that describe the $CO_2$ concentration in the indoor air in the context of presented methodology. Last section concludes the paper.


\section{Model}
\label{sec:model}
{Here we present continuous time random walk (CTRW) methodology \cite{bib:Montrol1965} first introduced for description of the motion of a particle which arrives at a position, waits random time and then next jumps randomly to the new position. Thus such a motion is determined by the waiting time and jump distributions. CTRWs models are nowadays well established and popular mathematical objects particularly attractive in the description of the so called anomalous diffusion phenomenon \cite{metzler2000,bib:Klafter2011}. It is an interesting observation that many CTRWs converge under suitable assumptions to the so called subordinated processes i.e. processes where physical time is replaced by some another stochastic process \cite{bib:Magdziarz2009,bib:Magdziarz22009}. For instance CTRW model with power law heavy-tailed stable waiting times and finite second moment jump distribution converges to the subordinated Brownian motion \cite{bib:MeerschaertInfinite} $B(S_\alpha^{-1}(t))$. Here $B(\cdot)$ is a Brownian motion and $S_\alpha^{-1}(\cdot)$ is the so called inverse $\alpha$-stable subordinator defined in \eqref{inverseStable}.  However one can name many more examples of subordinated processes and their different applications, both regarding the outer and inner processes.  To name only few we mention subordinated fractional Brownian motion \cite{bib:Gajda2013}, subordinated $\alpha-$stable L\'evy process where as a time change authors considered tempered stable  subordinator, \cite{bib:AgnieszkaJanusz2014} and inverse gamma subordinated Brownian motion \cite{bib:Janczura2011}. }
In this paper we consider real data set of  discrete observations for both time and scale  hence it is reasonable to consider as a proper model the system based on continuous time random walk  methodology. In such classical CTRW setting one considers the following 
stochastic process
\begin{equation}
\label{CTRW_proc}
Y(t)=\sum\limits_{i=1}^{N(t)} X_i,
\end{equation}
where $N(t)$ is a a counting process defined as 
\begin{equation}
\label{counting}
N(t)=\max\left\{k\geq 0:\;\sum\limits_{i=1}^k T_i\leq t \right\}.
\end{equation}
Here $T_i\; i=1,2,...$ form a sequence of positive independent, identically distributed (IID) random variables which can be seen as a waiting times between consecutive jumps $X_i$. We assume that the sequences $\{T_i\}_{i=1}^\infty$ and $\{X_i\}_{i=1}^\infty$  form  independent sequences. 

{ 
In this paper we consider  three particular distributions of waiting times for process defined in (\ref{CTRW_proc}), namely $\alpha$-stable, tempered stable and gamma. We denote such random variables as $T_i^S, T_i^T$ and $T_i^G$, respectively. In the case when the sequence $\{T_i^S\}$ constitutes a sample of IID random variables with $\alpha$-stable distribution, for each $i$ the random variable $T_i^S$ has the following Laplace transform \cite{bib:JanickiWeron}
\begin{equation}
\label{stable}
\mathbb{E}\left( e^{-zT_i^S}\right)=e^{-\sigma^\alpha z^\alpha},
\end{equation}
where $\sigma>0$ is a scale parameter and $\alpha\in(0,1)$.
In case of tempered stable waiting times, the $T_i^T$ random variable has the following Laplace transform \cite{bib:MagdziarzGajda2010,bib:Meerschaert2010}
\begin{equation}
\label{tempered}
\mathbb{E}\left( e^{-zT_i^T}\right)=e^{c(\lambda^\alpha-(\lambda+z)^\alpha)},
\end{equation}
for some parameters $c,\lambda>0$ and $\alpha\in (0,1)$. One observes that taking $\lambda=0$ tempered law reduces to stable one. Tempered stable distributions are particularly important in applications due to the fact that they can be in the same time close to $\alpha$-stable distributions and posses finite moments of all orders. From the above Laplace transform one can infer relation between probability density functions (PDFs) of pure 
$\alpha$-stable ($g_{T^S}$) and tempered stable ($g_{T^T}$) distributions, namely \cite{bib:Janczura2013}
\begin{equation*}
g_{T^T}(x)=e^{-\lambda x+\lambda^\alpha}g_{T^S}(x).
\end{equation*}
Third class of distribution of waiting times is the positive gamma distribution for which Laplace transform has the form \cite{bib:Janczura2011}
\begin{equation}
\label{gamma}
\mathbb{E}\left( e^{-zT_i^G}\right)=\left(\frac{1}{1+z\theta}\right)^k,\; \theta>0, k>0.
\end{equation}
In this paper we assume that jumps possess finite second moment thus applying classical Donsker theorem \cite{bib:Whitt} one can prove their convergence to Brownian motion. The same argument applies also when waiting times have finite second moment. Thus this holds in case of tempered stable and gamma waiting times in CTRW scenario. However, when the sequence $T_i^S$ belongs to the domain of attraction of one sided L\'evy stable distribution (defined via Laplace transform \eqref{stable}) then \cite{bib:MeerschaertInfinite}
\begin{equation*}
n^{-1/\alpha}\sum\limits_{i=1}^{[nt]}T_i^S{\overset{d}{\to}} S_\alpha(t)
\end{equation*}
as $n\to\infty$ for fixed $t$. Here $S_\alpha(t)$ is the $\alpha$-stable subordinator with the Laplace transform \cite{bib:Janczura2011}
\begin{equation*}
\mathbb{E}\left( e^{-sS_\alpha(t)}\right)=e^{-ts^\alpha}.
\end{equation*}
The above notation "${\overset{d}{\to}}$" means "convergence in distribution". In the context of the CTRW convergence let us also introduce the so called inverse $\alpha$-stable subordinator $S_\alpha^{-1}(t)$ defined as a first passage time of $S_\alpha(t)$ i.e. 
\begin{equation}
\label{inverseStable}
S_\alpha^{-1}(t)=\inf\{\tau\geq 0:S_\alpha(\tau)>t\}.
\end{equation}
Then one can formulate the following simple fact for heavy-tailed CTRWs. }
\begin{prop}
Let $Y(t)$ be the CTRW process defined in \eqref{CTRW_proc}. Assume that random variables $X_i$ are IID with finite mean $0$ and second moment equal to $1$. Moreover, the jump times $T_i^S$ belong to the domain of attraction of one sided $\alpha-$stable distribution \eqref{stable} with some $\alpha\in(0,1)$. Then 
\begin{equation}
    \label{firstCTRW}
	\frac{Y(nt)}{n^{\alpha/2}}{\overset{d}{\to}}B(S_\alpha^{-1}(t)),
\end{equation}
where $B(\cdot)$ is a standard Brownian motion and $S_\alpha^{-1}(t)$ is inverse $\alpha$-stable subordinator defined in \eqref{inverseStable}.
\end{prop}
Proof of this statement is an consequence of the Theorem 1 in \cite{bib:Zebrowski2012}.\\




\section{Estimation Procedure}
\label{sec:estimation}

The estimation procedure of the parameters corresponding to the considered CTRW model defined in (\ref{CTRW_proc}) is divided into few steps. One of the specific behaviour of the CTRW process (and corresponding processes driven by inverse subordinator) is the fact that they exhibit so called "trapping events" behaviour, i.e. visible constant time periods   (the time-intervals for which the process stays on the same level). Here we will use this property.  Therefore in the first step in the estimation scheme we divide the analysed vector of observations into two vectors. The first one is related to the lengths of "trapping events" while the second one represents the vector of observations that arises after removing the "trapping events". This standard procedure of analysis of CTRW models was used in various applications, see for example \cite{bib:Maciejewska2012,bib:Wylomanska2016,bib:Gajda2012_2}. The   parameters corresponding to the waiting times distribution we estimate on the basis of lengths of "trapping events". The whole procedure of testing and estimation of the waiting times distribution is presented in the further part of this section. After estimation of the waiting times distribution we analyse the process after removing the "trapping events", i.e jumps in the CTRW scenario. The details we present in the next subsections.
\subsection{Waiting times analysis}\label{sec:waitingTimesEstimation}

One of the main questions that arise during estimation procedure of the waiting times is identification of their correct distribution. In the literature most often it is assumed that the observations related to waiting times  follow infinitely divisible distribution e.g. $\alpha-$stable, tempered stable or gamma \cite{bib:Janczura2011}. The  most commonly used method to identify the correct distribution is visualization of empirical cumulative distribution function (CDF) on a plot with log-log scale and comparison with fitted parametric distribution \cite{bib:Wylomanska2012}. However, a testing procedures which would allow to identify the correct class of distribution for observations corresponding to waiting times are very rarely considered in the literature. The difficulty arises from the discretization of the time intervals which is a natural consequence of time-discretization of measuring devices. Furthermore, standard estimation procedures like method of moments in case of gamma and tempered stable distribution \cite{bib:Gajda2013} and McCulloch or regression method in case of $\alpha-$stable distribution \cite{bib:Burnecki2012} might lead to under- or overestimation of parameters. The other methods, such as fitting of linear and exponential functions in a log-log scale, respectively for $\alpha-$stable and tempered stable distributions, require setting arbitrary thresholds which usually depend on the set of parameters \cite{bib:Orzel2011}.

In this paper we propose a new estimation procedure which can be applied for fitting distribution function when the observations are rounded or discretized, e.g. waiting times until a characteristic of a process changes. In addition, we show that the procedure can be used to identify the class of distribution which best fits to the data among certain classes. As an example we perform a comparison between $\alpha-$stable, tempered stable and gamma distribution, as the most commonly used for description of waiting times distribution.

The main issue during estimation of constant time periods comes from the fact that the exact waiting time is unknown and usually comes from continuous distribution. For example, if we observe that a character of the process has changed after 3 units of time, it is not known at which point of time the change actually happened, the correct value lies in the interval $(2, 4)$ which can be seen in Fig. \ref{fig:periodComparison}. In other words, a constant time period equal to 2.5 might be classified as 2 or 3 with equal probability.
\begin{figure}
\centering
  \includegraphics[width=0.5\textwidth]{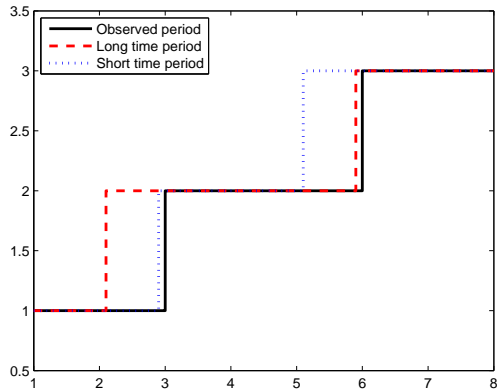}
\caption{Comparison of potential real time period versus observed value. Although observed value is 3 time units, the exact value might be short time period close to 2 as well as long time period close to 4.}
\label{fig:periodComparison}
\end{figure}

Due to these facts we introduce a modified version of cumulative distribution function. Let X be a non-negative continuous random variable with cumulative distribution function $F$ and probability distribution function $f$. We define mass function $\tilde{f}$ with a support on natural numbers in the following way
\begin{align*}
&\tilde{f}(0) = \int\limits_{0}^{1} f(x)(1-x) dx  \\
&\tilde{f}(n) = \int\limits_{n-1}^{n} f(x)(x-n+1) dx + \int\limits_{n}^{n+1} f(x)(n-x+1) dx \quad \mbox{for $n \geq 1$.}
\end{align*}
The modified cumulative distribution function is defined as $\tilde{F}(x) = \sum_{i \leq x} \tilde{f}(i)$. At the points of discontinuity the function can be expressed as
\begin{align*}
&\tilde{F}(0) = \int\limits_{0}^{1} f(x)(1-x) dx  \\
&\tilde{F}(n) = \hat{F}(n-1) + \int\limits_{n-1}^{n} f(x)(x-n+1) dx + \int\limits_{n}^{n+1} f(x)(n-x+1) dx \quad \mbox{for $n \geq 1$,} \numbereq
\label{mCDF}
\end{align*}
where $n \in \mathbb{N}$. Since for waiting times we limit the classes of distribution functions only to the ones with non-negative support and assume that the observations are done in equal time intervals, it is sufficient to define $\tilde{F}$ on the set of natural numbers. However, it is straightforward to extend the definition of the function to a countable set.

\begin{figure}
\centering
  \includegraphics[width=0.5\textwidth]{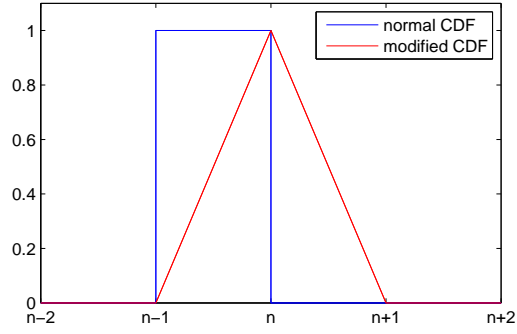}
\caption{Difference in integration between cumulative and  modified cumulative distribution functions.}
\label{fig:contrast}
\end{figure}

The graphical interpretation of differences between  cumulative and introduced modified cumulative distribution functions is illustrated in Fig. \ref{fig:contrast}. Furthermore, it is easy to show that
\begin{align*}
\tilde{F}(n) &= F(n) + \int_{n}^{n+1} f(x)(n-x+1) dx = \\
&= F(n) + (n+1)(F(n+1)-F(n)) -\int_{n}^{n+1} xf(x) dx = \\
&= (n+1)F(n+1)-nF(n) - ((n+1)F(n+1) - nF(n) - \int_{n}^{n+1} F(x) dx) = \\
&= \int_{n}^{n+1} F(x) dx.
\end{align*}
\begin{prop}
The modified cumulative distribution function $\tilde{F}$ can be considered as cumulative distribution function.
\end{prop}
\begin{proof}
It is easy to show that modified cumulative distribution function $\tilde{F}$ is bounded between 0 and 1, namely
\begin{equation*}
0 = \int_{n}^{n+1} 0 \leq \int_{n}^{n+1} F(x) \leq \int_{n}^{n+1} 1 = 1.
\end{equation*}
Straightforward calculations show that
\begin{equation*}
\lim\limits_{n\rightarrow\infty}{\tilde{F}(n)} = \lim\limits_{n\rightarrow\infty}{F(n) + \int_{n}^{n+1} f(x)(n-x+1) dx} \geq \lim\limits_{n->\infty}{F(n)dx} = 1
\end{equation*}
since the upper boundary is equal to 1, we conclude that $\lim\limits_{n\rightarrow\infty}{\tilde{F}(n)} = 1$. Since the support of mass function are natural numbers, it is obvious that $\lim\limits_{n\rightarrow-\infty}{\tilde{F}(n)} = 0$.
Moreover, we can show the modified cumulative distribution function is non-decreasing
\begin{equation*}
\tilde{F}(n+1)-\tilde{F}(n) = \int_{n}^{n+1} F(x) dx - \int_{n-1}^{n} F(x) dx \geq \int_{n}^{n+1} F(n) dx - \int_{n-1}^{n} F(n) dx = 0.
\end{equation*}
At the end we mention, the modified cumulative distribution function $\tilde{F}$ is right-continuous which stems from its definition.
\end{proof}
One of the main goal of this paper is to introduce a new method of estimation  of waiting times distribution parameters under the assumption that they are continuously distributed. Moreover, we propose also the procedure of waiting times distribution recognition among chosen distributions. Due to the fact that between consecutive measurements only a single change in the process can be observed, a further modification of the function $\tilde{F}$ defined in (\ref{mCDF}) is needed. We define $G$ as rescaled modified cumulative distribution function (RMCDF) in the following way
\begin{align*}
&G(0) = 0 \\
&G(n) = \frac{\tilde{F}(n)}{1-\tilde{F}(0)} \quad \mbox{for $n \geq 1$.} \numbereq
\label{finalCDF}
\end{align*}
It is obvious that $G$ still has properties of cumulative distribution function. The comparison between both types of cumulative distribution functions can be seen in Fig. \ref{fig:cdfComp}. The distance between empirical CDF of waiting times can be easily observed: on the left panel the theoretical CDF of $\alpha-$stable distribution underestimates the empirical one whereas on the right panel the theoretical CDF of tempered stable distribution overestimates the empirical one. In both cases the rescaled modified CDF provides reasonable fit to the data.

\begin{figure}
\centering
  \includegraphics[width=1\textwidth]{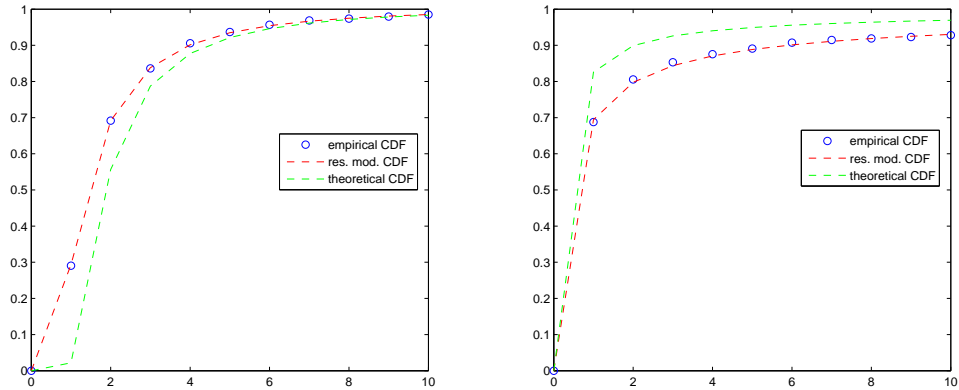}
\caption{Comparison of empirical (based on the waiting times observations), theoretical and rescaled modified cumulative distribution functions. The observations were taken from a single trajectory of CTRW of length 10000, waiting times simulated using: left panel - stable distribution with $\alpha = 0.7, \sigma=0.1$, right panel - tempered stable distribution with $\alpha=0.8, \lambda=0.1, \sigma=2$.}
\label{fig:cdfComp}
\end{figure}

The proposed methodology of fitting the distribution parameters corresponding to waiting times is based on the minimum distance estimation applied to a certain type of distribution. Let $K$ and $L$ denote two functions with a common support on $\mathbb{R}$, the considered distances are
\begin{itemize}
	\item Kolmogorov-Smirnov (KS)
	\begin{eqnarray*}
	KS(K,L) = \sup_{x \in R}{\left|K(x)-L(x)\right|}
	\end{eqnarray*}
	\item Cram\'er-von Mises (CvM)
	\begin{eqnarray*}
	CvM(K,L) = \int_{-\infty}^{\infty}(K(x)-L(x))^2 dL(x)
	\end{eqnarray*}
	\item Anderson-Darling (AD)
	\begin{eqnarray*}
	AD(K,L) = \int_{-\infty}^{\infty} \frac{(K(x)-L(x))^2}{L(x)(1-L(x))} dL(x).
	\end{eqnarray*}
\end{itemize}
In our estimation procedure we consider the distance between the rescaled modified cumulative distribution function introduced in (\ref{finalCDF}) and empirical distribution function of waiting times defined by
\begin{eqnarray}
\hat{F}_n(t) = \frac{1}{n}\sum_{i=1}^{n}{\textbf{1}_{x_i\leq t}},
\label{empCDF}
\end{eqnarray}
where ${\textbf{1}_{A}}$ is the indicator of the set $A$. For each class of cumulative distribution functions  we can find parameters which minimize the distance to the empirical distribution function, i.e. $G_{\theta_0}$ which satisfies the following condition
\begin{eqnarray}
D(G_{\theta_0}, \hat{F}) = \inf_{\theta \in \Theta}{D(G_{\theta}, \hat{F})},
\label{minDistance}
\end{eqnarray}
where $D$ is one of the introduced distances: $KS$, $CvM$ and $AD$, $\hat{F}$ is empirical CDF introduced in (\ref{empCDF}), $G$ is a rescaled cumulative distribution function defined in (\ref{finalCDF}) and $\Theta$ is the set of parameters of a certain class of distribution functions. As it was mentioned, in this paper we consider three classes of distributions, namely $\alpha-$stable, tempered stable and gamma and for the selected distributions the estimation was done for the following parameter spaces
\begin{itemize}
	\item $\alpha$-stable:
	$\alpha \in (0,1), \sigma \in \mathbb{R}_{+}$
	\item tempered stable:
	$\alpha \in (0,1), \lambda \in \mathbb{R}_{+}, \sigma \in \mathbb{R}_{+}$
	\item gamma: 
	$k \in \mathbb{R}_{+}, \theta \in \mathbb{R}_{+}$.
\end{itemize}

Due to discretization of the input data the uniqueness of the solution cannot always be guaranteed. However, by using standard estimation procedures as an initial guess and by applying reasonable boundary conditions, we have shown in our simulation study that in almost all cases the convergence to the true parameters was fulfilled. 
For finding the minimum distance we have used the Nelder-Mead simplex algorithm described in \cite{bib:Lagaris1998}.

It is worth mentioning that the proposed methodology can be applied to broad set of distributions - in general it may serve as a method of fitting any distribution function if the observations are rounded. However, as it was mentioned, in this paper we have limited the classes of waiting times distributions to $\alpha-$stable, tempered stable and gamma distributions. 

\subsection{Jumps analysis}
\label{sec:jumpsEstimation}

The main focus of this paper is the estimation of waiting times in CTRW scenario, however for the sake of completeness we mention as well estimation procedure of jumps parameters. The problem has been widely discussed in multiple articles. In particular, in \cite{bib:Janczura2011} it was proposed to make the probabilities of jumps dependent on a deterministic periodic function. In case of carbon dioxide concentration (see section \ref{sec:analysis}) the periodicity is not observed.

Due to specific properties of the real data analyzed in the next section we have proposed a simplified estimation of probabilities of jumps which follow the process 

\begin{eqnarray}
X_t = \left\{
                \begin{array}{ll}
                 a	\quad \textrm{with probability $p_t$} \\
								-a \quad \textrm{with probability $1-p_t$.} 
             \end{array}
          \right.
\label{jumpsEstimation}
\end{eqnarray}

The procedure of estimation the $p_t$ function is as follows: owing to the fact that to the analysis we took the data corresponding to one day of the week (Monday) and we consider all taken time series as realization of the same process, then $p_t$ is just a number of upside jumps divided by the number of considered trajectories for given time point $t$. The value of $a$ corresponds to the resolution of the sensor and in our case is equal to $50$. The similar procedure was proposed in \cite{bib:Janczura2011} also for indoor air quality data. 


\section{Simulation Study}
\label{sec:simulation}

In the simulation study we have used 100 samples, each consisting of 3000 observations. Simulation schema of CTRW model (and corresponding subordinated process) can be found for example in \cite{sub3}. The decision regarding length of samples was driven by initial analysis of the data, see section \ref{sec:analysis}. We assume the jumps $X_i$ in the considered model (\ref{CTRW_proc}) are independent identically distributed random variables equal $50$ or $-50$ with probabilities $0.5$. Moreover we consider three waiting times distributions, namely $\alpha-$stable, tempered stable and gamma.

At first we test our algorithm to identify correct distribution class corresponding to waiting times. After extraction of the waiting times from the considered trajectory we ran the algorithm to find the values of the parameters which minimize the distance in each of the distribution classes. The values of the estimators with the smallest distance  indicate the distribution of the best fitting to considered data set. The distribution from which the parametric estimator has been chosen is then marked as the most appropriate for data description.

\begin{table}[]
\centering
\caption{Percentage of correctly identified distribution classes. Average distance is given in the bracket.}
\label{percentageTable}
\begin{tabular}{cc|c|c|c|}
\cline{3-5}
\multicolumn{1}{l}{}                                   &                  & KS          & CvM         & AD          \\ \hline
\multicolumn{1}{|c|}{\multirow{3}{*}{stable}}          & S(0.6, 1)        & 68 (0.1682) & 74 (0.0006) & 72 (0.0331) \\ \cline{2-5} 
\multicolumn{1}{|c|}{}                                 & S(0.75, 1)       & 81 (0.0829) & 91 (0.0004) & 91 (0.0025) \\ \cline{2-5} 
\multicolumn{1}{|c|}{}                                 & S(0.9, 1)        & 93 (0.0526) & 99 (0.0002) & 99 (0.0011) \\ \hline
\multicolumn{1}{|c|}{\multirow{3}{*}{tempered stable}} & TS(0.3, 1, 1)    & 98 (0.0019) & 100 (0.0001) & 100 (0.0001) \\ \cline{2-5} 
\multicolumn{1}{|c|}{}                                 & TS(0.6, 0.5, 1)  & 97 (0.0021) & 100 (0.0001) & 100 (0.0001) \\ \cline{2-5} 
\multicolumn{1}{|c|}{}                                 & TS(0.9, 0.1, 1)  & 99 (0.0014) & 100 (0.0001) & 100 (0.0001) \\ \hline
\multicolumn{1}{|c|}{\multirow{3}{*}{gamma}}           & G(0.5, 10)       & 93 (0.0181) & 100 (0.0001) & 100 (0.0006) \\ \cline{2-5} 
\multicolumn{1}{|c|}{}                                 & G(2, 5)          & 96 (0.0235) & 98 (0.0002) & 98 (0.0011) \\ \cline{2-5} 
\multicolumn{1}{|c|}{}                                 & G(5, 2)          & 97 (0.0263) & 99 (0.0002) & 99 (0.0015) \\ \hline
\end{tabular}
\end{table}

The results presented in Table \ref{percentageTable} indicate that out of the three distances the most accurate in selecting proper distribution were $CvM$ and $AD$. Lower efficiency of the algorithm in the $\alpha$-stable distribution is due to low number of observations.

In the next step we have examined the robustness of our estimators with respect to various combinations of distribution parameters. The dispersion from the real parameter has been measured using mean square error, for a sample consisting of $n$ observations the error is calculated in the following way
\begin{eqnarray}
MSE(\hat{\theta}) = \sum^{n}_{i=1}(\hat{\theta}_i - \theta)^2.
\label{biasMSE}
\end{eqnarray}
In addition, we have compared our estimation procedure with the most commonly used estimators McCulloch in case of $\alpha$-stable distribution and method of moments in case of tempered stable and gamma distributions. It is worth mentioning that in case of $\alpha$-stable distribution we have tested as well the regression method. However, the results were worse than in case of McCulloch method therefore for the final analysis only results of the latter method are shown. 
\begin{table}[hbtp]
\centering
\caption{Estimation of $\alpha$-stable distribution parameters. In each cell two values are given: upper one corresponds to median of estimated parameter, the lower one is MSE. The notation used is S($\alpha, \sigma$). The mean squared error is given in scale $10^{-3}$.}
\label{alphaStableEstim}
\begin{tabular}{c|c|c|c|c|c|c|c|c|}
\cline{2-9}
                                   & \multicolumn{4}{c|}{$\alpha$}                       & \multicolumn{4}{c|}{$\sigma$}                       \\ \cline{2-9} 
                                   & KS              & CvM             & AD              & McC              & KS              & CvM             & AD              & McC              \\ \hline
\multicolumn{1}{|c|}{S(0.6, 0.5)}  & \specialcell{0.759\\26.0} & \specialcell{0.626\\7.0} & \specialcell{0.626\\7.2} & \specialcell{0.650\\80.2} & \specialcell{0.550\\42.5} & \specialcell{0.535\\41.9} & \specialcell{0.539\\24.0} & \specialcell{0.733\\53689249.6} \\ \hline
\multicolumn{1}{|c|}{S(0.6, 1)}    & \specialcell{0.744\\25.1} & \specialcell{0.623\\5.6} & \specialcell{0.620\\6.5} & \specialcell{0.634\\89.3} & \specialcell{0.748\\180.1} & \specialcell{1.008\\71.2} & \specialcell{0.999\\130.2} & \specialcell{1.266\\53791999.7} \\ \hline
\multicolumn{1}{|c|}{S(0.6, 2)}    & \specialcell{0.775\\29.6} & \specialcell{0.632\\7.8} & \specialcell{0.628\\6.7} & \specialcell{0.641\\147.2} & \specialcell{1.083\\970.4} & \specialcell{1.756\\444.6} & \specialcell{1.778\\350.2} & \specialcell{2.309\\100341231.4} \\ \hline
\multicolumn{1}{|c|}{S(0.75, 0.5)} & \specialcell{0.807\\3.9} & \specialcell{0.752\\0.8} & \specialcell{0.753\\0.9} & \specialcell{0.746\\34.6} & \specialcell{0.437\\7.0} & \specialcell{0.498\\14.9} & \specialcell{0.496\\3.6} & \specialcell{0.526\\24657038.0} \\ \hline
\multicolumn{1}{|c|}{S(0.75, 1)}   & \specialcell{0.816\\6.0} & \specialcell{0.752\\1.5} & \specialcell{0.752\\1.5} & \specialcell{0.765\\40.0} & \specialcell{0.777\\77.5} & \specialcell{0.992\\36.9} & \specialcell{0.999\\31.1} & \specialcell{0.989\\25456771.7} \\ \hline
\multicolumn{1}{|c|}{S(0.75, 2)}   & \specialcell{0.845\\10.9} & \specialcell{0.763\\2.8} & \specialcell{0.761\\2.5} & \specialcell{0.762\\95.9} & \specialcell{1.282\\653.0} & \specialcell{1.893\\249.6} & \specialcell{1.909\\171.8} & \specialcell{2.153\\43178539.9} \\ \hline
\multicolumn{1}{|c|}{S(0.9, 0.5)}  & \specialcell{0.918\\1.3} & \specialcell{0.903\\0.2} & \specialcell{0.901\\0.2} & \specialcell{0.959\\37.7} & \specialcell{0.480\\12.4} & \specialcell{0.490\\5.3} & \specialcell{0.493\\2.5} & \specialcell{0.756\\24459941.5} \\ \hline
\multicolumn{1}{|c|}{S(0.9, 1)}    & \specialcell{0.920\\2.3} & \specialcell{0.901\\0.1} & \specialcell{0.902\\0.1} & \specialcell{0.895\\42.6} & \specialcell{0.812\\136.6} & \specialcell{0.975\\23.2} & \specialcell{0.976\\20.5} & \specialcell{0.941\\24477120.4} \\ \hline
\multicolumn{1}{|c|}{S(0.9, 2)}    & \specialcell{0.940\\4.0} & \specialcell{0.902\\0.2} & \specialcell{0.902\\0.5} & \specialcell{0.900\\49.8} & \specialcell{1.281\\1151.4} & \specialcell{1.926\\95.8} & \specialcell{1.940\\81.8} & \specialcell{1.960\\461.5} \\ \hline
\end{tabular}
\end{table}

First, we examine the correctness of proposed estimator for $\alpha$-stable distribution. For various configurations of parameters we estimate $\alpha$ and $\sigma$. It is easy to see that the mean squared error compared to McCulloch method is significantly smaller. The performance of the method increases with higher $\alpha$ and smaller values of $\sigma$. It is obvious and stems from the fact that more observations are used for estimation procedure, e.g. for $\alpha = 0.6, \sigma = 2$ on average for a single trajectory there were only 57 waiting times while respectively for $\alpha=0.9, \sigma=0.5$ there were 407 waiting times. 

Out of the three proposed distances Cramer-von Mises and Anderson-Darling have comparable results and in almost all cases are more efficient than Kolmogorov-Smirnov. 

In case of tempered stable distribution we have examined changes with respect to $\alpha, \lambda$ and $\sigma$. Since the changes with respect to $\sigma$ resemble the ones for $\alpha$-stable distribution we present only the changes in terms of $\alpha$ and $\lambda$. 

\begin{table}[]
\centering
\caption{Tempered stable distribution parameter estimation. The upper parameter corresponds to median, the lower one is MSE. The notation used is TS($\alpha, \lambda, \sigma$). The mean squared error values are given in scale $10^{-3}$.}
\label{temperedStableTable}
\begin{tabular}{c|c|c|c|c|c|c|c|c|}
\cline{2-9}
\multicolumn{1}{l|}{}                 & \multicolumn{4}{c|}{$\alpha$}                       & \multicolumn{4}{c|}{$\lambda$}                      \\ \cline{2-9} 
\multicolumn{1}{l|}{}                 & KS              & CvM             & AD              & MoM              & KS              & CvM             & AD              & MoM              \\ \hline
\multicolumn{1}{|c|}{TS(0.3, 0.1, 1)} & \specialcell{0.294\\0.6} & \specialcell{0.277\\6.3} & \specialcell{0.278\\7.1} & \specialcell{0.471\\39.3} & \specialcell{0.109\\0.7} & \specialcell{0.099\\0.2} & \specialcell{0.101\\0.2} & \specialcell{0.098\\0.7} \\ \hline
\multicolumn{1}{|c|}{TS(0.3, 0.5, 1)} & \specialcell{0.300\\0.2} & \specialcell{0.295\\0.7} & \specialcell{0.278\\4.2} & \specialcell{0.749\\205.1} & \specialcell{0.520\\2.4} & \specialcell{0.513\\3.0} & \specialcell{0.509\\2.9} & \specialcell{0.361\\22.1} \\ \hline
\multicolumn{1}{|c|}{TS(0.3, 1, 1)}   & \specialcell{0.318\\0.7} & \specialcell{0.299\\0.1} & \specialcell{0.342\\3.1} & \specialcell{0.876\\331.4} & \specialcell{1.045\\9.4} & \specialcell{1.057\\10.6} & \specialcell{1.043\\13.0} & \specialcell{0.578\\187.3} \\ \hline
\multicolumn{1}{|c|}{TS(0.6, 0.1, 1)} & \specialcell{0.595\\0.7} & \specialcell{0.692\\4.3} & \specialcell{0.589\\4.1} & \specialcell{0.644\\6.2} & \specialcell{0.117\\2.0} & \specialcell{0.102\\0.7} & \specialcell{0.106\\0.6} & \specialcell{0.112\\1.1} \\ \hline
\multicolumn{1}{|c|}{TS(0.6, 0.5, 1)} & \specialcell{0.598\\0.3} & \specialcell{0.594\\1.3} & \specialcell{0.530\\6.4} & \specialcell{0.803\\42.4} & \specialcell{0.513\\2.3} & \specialcell{0.500\\2.9} & \specialcell{0.531\\5.8} & \specialcell{0.405\\14.8} \\ \hline
\multicolumn{1}{|c|}{TS(0.6, 1, 1)}   & \specialcell{0.612\\0.7} & \specialcell{0.604\\0.5} & \specialcell{0.585\\3.2} & \specialcell{0.885\\82.7} & \specialcell{1.033\\9.0} & \specialcell{1.025\\10.4} & \specialcell{1.052\\15.0} & \specialcell{0.641\\131.6} \\ \hline
\multicolumn{1}{|c|}{TS(0.9, 0.1, 1)} & \specialcell{0.882\\0.7} & \specialcell{0.885\\0.8} & \specialcell{0.893\\1.0} & \specialcell{0.888\\0.6} & \specialcell{0.138\\5.5} & \specialcell{0.140\\3.3} & \specialcell{0.108\\3.2} & \specialcell{0.121\\3.1} \\ \hline
\multicolumn{1}{|c|}{TS(0.9, 0.5, 1)} & \specialcell{0.892\\0.2} & \specialcell{0.889\\0.2} & \specialcell{0.880\\2.6} & \specialcell{0.909\\0.2} & \specialcell{0.532\\4.8} & \specialcell{0.553\\3.7} & \specialcell{0.573\\22.4} & \specialcell{0.499\\9.3} \\ \hline
\multicolumn{1}{|c|}{TS(0.9, 1, 1)}   & \specialcell{0.903\\0.4} & \specialcell{0.897\\0.3} & \specialcell{0.897\\0.9} & \specialcell{0.930\\1.0} & \specialcell{1.029\\6.4} & \specialcell{1.014\\4.7} & \specialcell{1.120\\18.8} & \specialcell{0.808\\54.2} \\ \hline
\end{tabular}
\end{table}

At the end we examine changes with respect to shape and scale parameters for gamma distribution. In Table \ref{gammaTable} we have compared estimation of shape and scale parameters for proposed estimators and the method of moments.
In most of the cases minimum distance estimators outperform standard method of methods. Out of the three proposed estimation methods the one based on Kolmogorov-Smirnov distance gives the best results, i.e. the lowest mean squared error and relatively low bias for most of the simulated waiting times.
It is obvious that the estimation worsens when the number of observations decreases. For example, in case of gamma distribution with $k=5, \theta=10$ in a simulated trajectory of length 3000, on average only ca. 60 observations are used for estimation.

\begin{table}[]
\centering
\caption{Gamma distribution parameter estimation. The upper parameter corresponds to median, the lower one is MSE. The notation used is G($k, \theta$). The mean squared error values are given in scale $10^{-3}$.}
\label{gammaTable}
\begin{tabular}{c|c|c|c|c|c|c|c|c|}
\cline{2-9}
\multicolumn{1}{l|}{}                 & \multicolumn{4}{c|}{$k$}                       & \multicolumn{4}{c|}{$\theta$}                      \\ \cline{2-9} 
\multicolumn{1}{l|}{}                 & KS              & CvM             & AD              & MoM              & KS              & CvM             & AD              & MoM              \\ \hline
\multicolumn{1}{|c|}{G(0.5, 2)}    & \specialcell{0.523\\6.7} & \specialcell{0.503\\9.5} & \specialcell{0.510\\7.2} & \specialcell{1.614\\1266.1} & \specialcell{1.962\\21.9} & \specialcell{1.962\\42.1} & \specialcell{1.981\\33.2} & \specialcell{1.205\\642.6} \\ \hline
\multicolumn{1}{|c|}{G(0.5, 5)}    & \specialcell{0.517\\4.9} & \specialcell{0.497\\6.0} & \specialcell{0.505\\5.0} & \specialcell{0.974\\230.6} & \specialcell{4.974\\190.3} & \specialcell{5.049\\279.3} & \specialcell{5.011\\211.9} & \specialcell{3.817\\1470.8} \\ \hline
\multicolumn{1}{|c|}{G(0.5, 10)}   & \specialcell{0.516\\8.4} & \specialcell{0.503\\8.5} & \specialcell{0.502\\7.4} & \specialcell{0.776\\85.1} & \specialcell{9.626\\1378.7} & \specialcell{10.035\\2346.4} & \specialcell{9.939\\1796.5} & \specialcell{8.320\\3478.3} \\ \hline
\multicolumn{1}{|c|}{G(2, 2)}      & \specialcell{1.993\\2.7} & \specialcell{2.016\\18.9} & \specialcell{2.011\\17.1} & \specialcell{2.180\\56.3} & \specialcell{2.017\\5.1} & \specialcell{1.970\\17.9} & \specialcell{1.990\\15.6} & \specialcell{1.886\\26.2} \\ \hline
\multicolumn{1}{|c|}{G(2, 5)}      & \specialcell{2.061\\15.2} & \specialcell{2.032\\19.3} & \specialcell{2.038\\18.9} & \specialcell{2.036\\40.9} & \specialcell{4.843\\139.9} & \specialcell{4.979\\164.7} & \specialcell{4.921\\153.7} & \specialcell{4.941\\222.8} \\ \hline
\multicolumn{1}{|c|}{G(2, 10)}     & \specialcell{2.156\\22.1} & \specialcell{2.101\\26.2} & \specialcell{2.099\\26.2} & \specialcell{2.001\\75.3} & \specialcell{8.455\\2801.7} & \specialcell{9.257\\1025.0} & \specialcell{9.233\\1113.8} & \specialcell{10.027\\2094.1} \\ \hline
\multicolumn{1}{|c|}{G(5, 2)}      & \specialcell{4.990\\1.7} & \specialcell{5.005\\31.1} & \specialcell{4.950\\31.6} & \specialcell{4.874\\235.2} & \specialcell{1.999\\3.1} & \specialcell{1.997\\7.5} & \specialcell{1.991\\6.5} & \specialcell{2.042\\41.2} \\ \hline
\multicolumn{1}{|c|}{G(5, 5)}      & \specialcell{5.135\\18.0} & \specialcell{5.144\\35.3} & \specialcell{4.989\\36.8} & \specialcell{4.766\\593.0} & \specialcell{4.525\\315.2} & \specialcell{4.766\\123.9} & \specialcell{4.812\\122.3} & \specialcell{5.278\\810.7} \\ \hline
\multicolumn{1}{|c|}{G(5, 10)}     & \specialcell{5.200\\40.0} & \specialcell{5.199\\39.0} & \specialcell{5.199\\38.6} & \specialcell{4.522\\1351.6} & \specialcell{7.701\\5944.7} & \specialcell{6.107\\8004.7} & \specialcell{6.636\\7172.4} & \specialcell{11.160\\8192.4} \\ \hline
\end{tabular}
\end{table}


\section{Carbon dioxide data analysis}
\label{sec:analysis}

In the following section we will perform analysis of carbon dioxide data. In subsection \ref{sec:dataDescription} we describe the data and introduce the data cleansing procedure.  In subsection \ref{sec:dataAnalysis} we describe the data using CTRW model with the main emphasis on waiting times analysis.

\subsection{Data Description}
\label{sec:dataDescription}

Measurements of carbon dioxide concentration were performed in a lecture room with an amphitheatric layout. Room dimensions are 19 x 8 m x (4 -2.9 m). It has only one external wall, which is fitted with huge, openable 6 windows. Despite availability of mechanical ventilation, air exchange is realized predominantly via the natural ventilation. Teaching hours extend from 9:00 to 21:00. Classes are held during all working days and on majority of weekends (part time studies). Teaching blocks are typically 1.5 h long with the brakes of 15 min in-between. Although designed for 90 students, the lecture room is hardly ever occupied to that extent. In the examined period, the number of listeners changed considerably within a single day as well as from day to day.

$CO_2$ measurements were performed on the central part of the room at the height of about 1 m. The measuring device was separated from the direct influence of the emission sources (students and the teacher). The monitoring was realized with the instrument dedicated to continuous measurements and data logging. It is based on the NDIR sensor and it offers the measuring range 0-5000 ppm; accuracy 50 ppm +3 \% of measured value and the measurement data resolution 1 ppm. This level of performance may be currently considered as a standard in indoor air quality studies.

The measurement results were recorded with time resolution of 15 s. The data was collected during 9 consecutive weeks, between 20th February and 22nd April 2013. Since the lecture hours are different for each day, in this paper we focus only on Mondays - the analysis can be easily replicated for other weekdays. The trajectories of eight consecutive Mondays is displayed in Fig. \ref{fig:consecMon}.

\begin{figure}
\centering
  \includegraphics[width=1\textwidth]{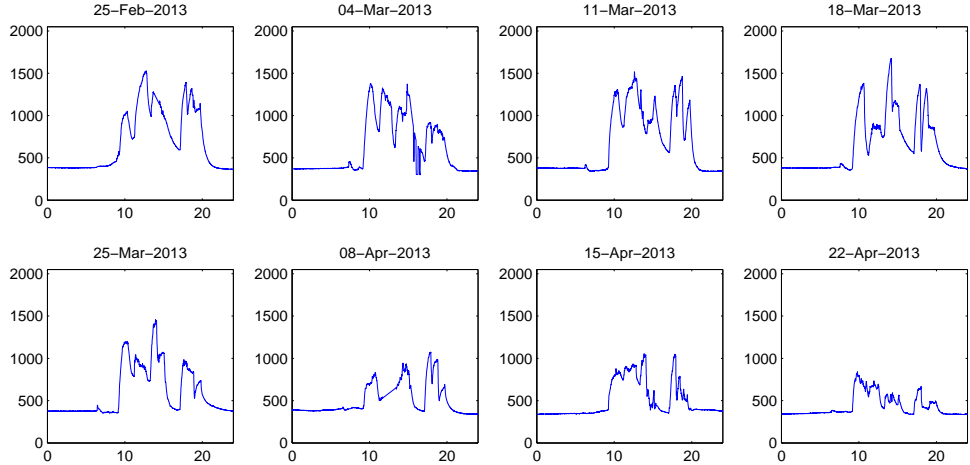}
\caption{Carbon dioxide concentration in eight consecutive Mondays from 25th February till 22nd April. Due to holidays, no lectures were held on 1st April.}
\label{fig:consecMon}
\end{figure}

By looking at the chart it is easy to notice that during lecture hours the carbon dioxide concentration rises significantly and decreases during the breaks. However, in contrast to the analysis of temperature data presented in \cite{bib:Janczura2013}, no particular trend can be observed in the data. Additionally, due to to irregular student attendance at lectures, estimation and prediction of $CO_2$ concentration based on the number of students subscribed to the lectures might be inaccurate.

Taking into account teaching hours and the measurement device accuracy we have performed 'cleansing' of the data in order to reveal only significant changes in the data:
\noindent
\begin{itemize}
\item The lectures are held only till 20:00. Around 21:00 the process stabilized around certain level and only small fluctuation can be observed till the next day lectures. We have called the time from 21:00 till 09:00 stale period and removed it from core analysis. As estimation of stale period is out of interest and the average value fluctuates around 400 ppm, we assume that during stale period the carbon dioxide concentration is equal to 400 ppm. 
\item The device accuracy is 50 ppm  $\pm$  3 $\%$ of the measured value. Such a difference between measured values indicates the change in the state of indoor air, with respect to $CO_2$ concentration, which is greater than the measurement error. Hence, it should be respected. For this reason we have rounded the observed values to the closest multiple of 50.
\item Furthermore, we have removed periods where the value has changed for short period - less than 30 seconds - and then came back to previous level.
\end{itemize}

\subsection{Data Analysis}
\label{sec:dataAnalysis}
\begin{figure}
\centering
  \includegraphics[width=1\textwidth]{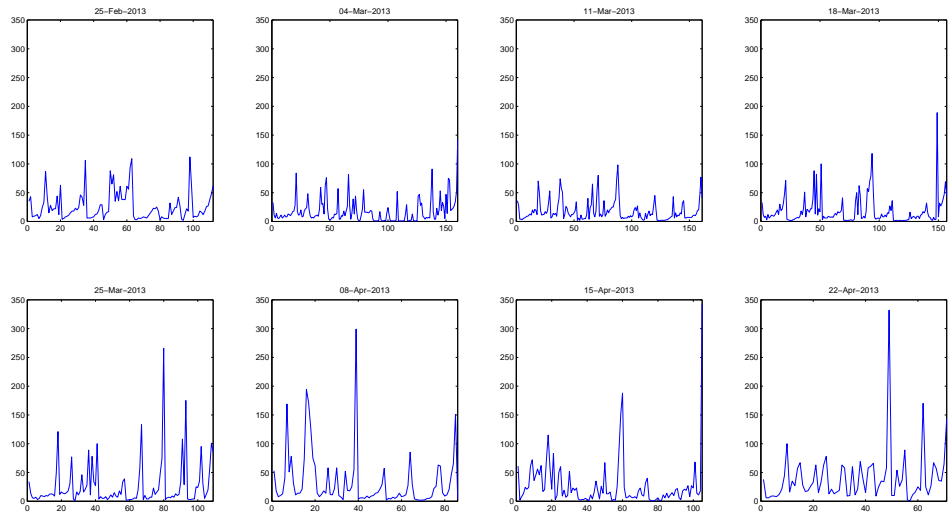}
\caption{The vectors of waiting times in eight consecutive Mondays from 25th February till 22nd April.}
\label{wait_aga}
\end{figure}
\begin{figure}
\centering
\includegraphics[width=1\textwidth]{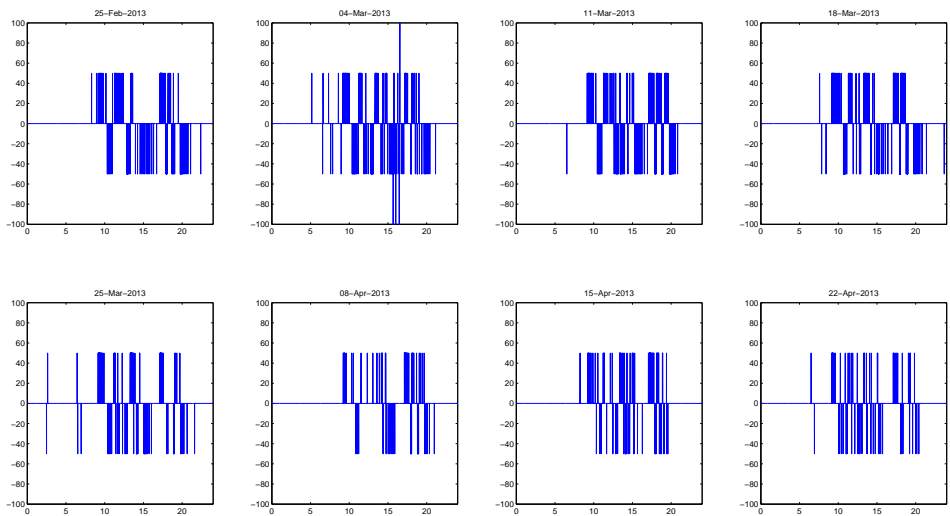}
\caption{The vectors of jumps in eight consecutive Mondays from 25th February till 22nd April. }
\label{jumps_aga}
\end{figure}
\begin{figure}
\centering
  \includegraphics[width=1\textwidth]{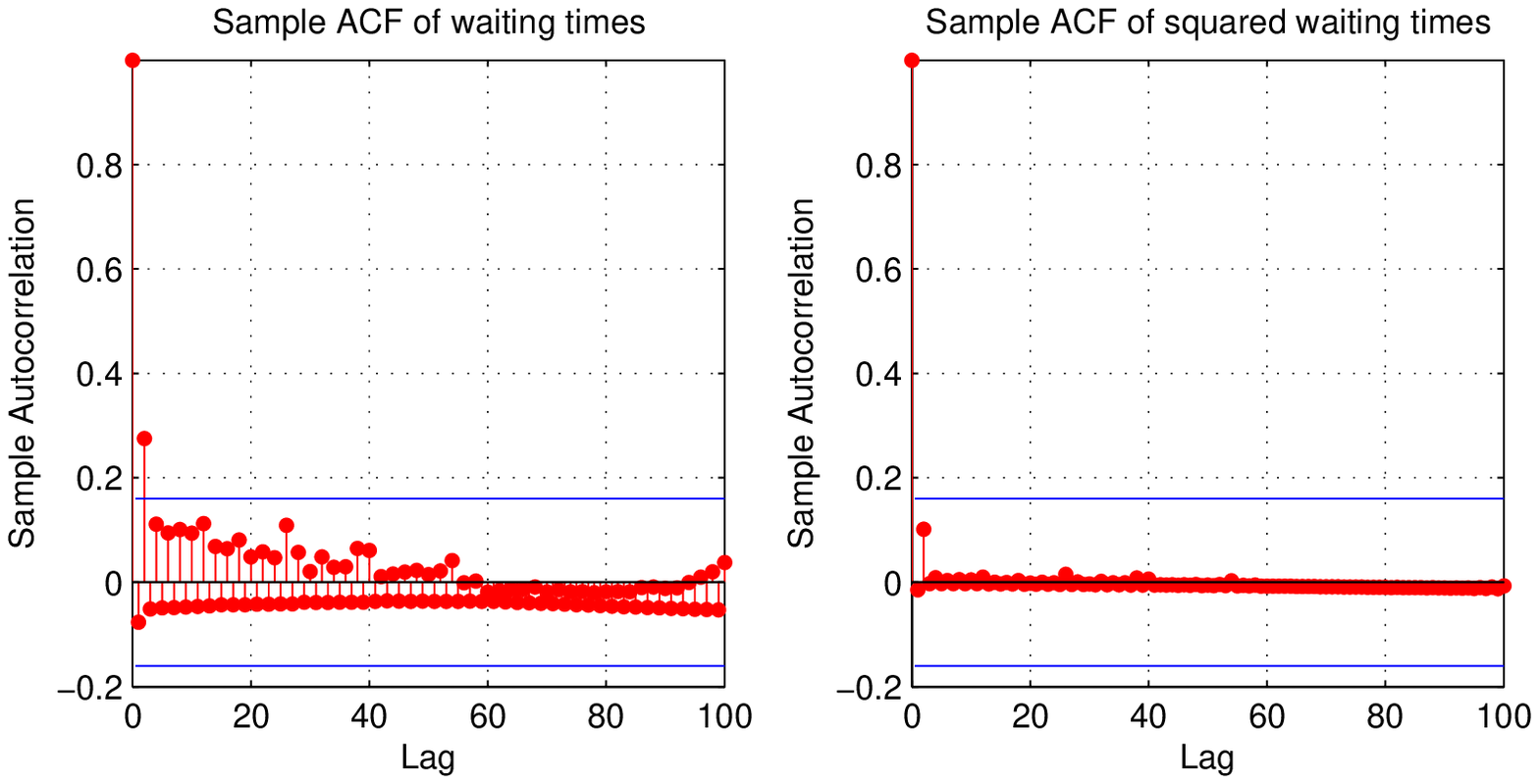}
\caption{The sample autocorrelation function of constant time periods series (left panel) and corresponding squared series (right panel) for 25th Feb.}
\label{fig:stationarityCheck}
\end{figure}

After applying the cleansing procedure described in subsection \ref{sec:dataDescription}, we considered CTRW as a stochastic system that allows modelling this type of time series. In general the cleansing procedure is not obligatory. In addition, the behaviour of raw data suggests that it might be modelled using CTRW as well.

At the first stage we have divided the analysed data into two vectors: first  corresponding to waiting times and second consisting of jump sizes. In Fig. \ref{wait_aga} we present the vectors representing the waiting times while in Fig. \ref{jumps_aga} - the corresponding jumps. 
We have assumed that the data constitutes of 8 realizations of the same stochastic process. Each sample has been analysed independently in the context of waiting times and probability of upward-downward movement.

In order to make use of CTRW to model the data, it needs to be validated if the vector $\{T_i\}$ of waiting times forms an independent, identically distributed sample. For each of 8 samples we test the independence with a visual check of the autocorrelation function of the series and squared series as described in \cite{bib:Wylomanska2012}. In Fig. \ref{fig:stationarityCheck} we present the obtained autocorrelation functions for the first sample (25th Feb) together with the confidence intervals for a white noise. As can be observed, the calculated values are close to 0 and most of them lie within the white noise confidence intervals. The procedure has been applied to all 8 samples and the results resemble the one in Fig. \ref{fig:stationarityCheck}. Hence, we may conclude that the waiting times are independent. The identity of distributions can be checked for example by testing the behaviour of the empirical second moment of the time series as described in \cite{bib:Gajda2012}.

The estimation of parameters has been performed using the method described in section \ref{sec:estimation}. As shown in our simulation study in section \ref{sec:simulation} the best fit was obtained using Anderson-Darling distance, hence this criteria was used for parameter estimation. In Fig. \ref{fig:cdfFittedToData}
we present fit of rescaled modified cumulative distribution function for all 3 distributions: $\alpha$-stable, tempered stable and gamma to empirical distribution function on 25th Feb. As shown in the plot, gamma distribution provides relatively good fit to the data and outperforms the other distributions.

Out of 8 samples in 12.5\% cases the best fit was obtained using $\alpha$-stable distribution, in 25\% cases tempered stable was chosen and in 62.5\% cases gamma. Furthermore, the differences between gamma distribution and the other ones in the remaining 37.5\% of cases was negligible. This indicated that gamma distribution is the most adequate for waiting times estimation. The results of fitting of gamma distribution are shown in Table \ref{MondaysResults}: the estimated values of shape parameter are similar for each sample, there is more variability in the estimated values of scale parameter. For further processing and simulation we have used the median of both parameters. Potentially a more robust estimation procedure would involve minimizing the distances of a particular distribution function with respect to all samples. On the other hand, such a method might be biased for example if one of the samples is not consistent with the others.

\begin{table}[]
\centering
\caption{Gamma distribution parameters estimated for 8 consecutive Mondays}
\label{MondaysResults}
\begin{tabular}{c|c|c|c|}
\cline{2-4}
\multicolumn{1}{c|}{}            & $k$    & $\theta$            &theoretical mean (in min.)         \\ \hline
\multicolumn{1}{|c|}{25.02.2013} & 0.3208 & 30.2587& 2.43 \\ \hline
\multicolumn{1}{|c|}{04.03.2013} & 0.4401 & 21.1849&2.33 \\ \hline
\multicolumn{1}{|c|}{11.03.2013} & 0.6294 & 11.0901 &1.75\\ \hline
\multicolumn{1}{|c|}{18.03.2013} & 0.4646 & 17.7034 &2.06\\ \hline
\multicolumn{1}{|c|}{25.03.2013} & 0.2326 & 32.0382 &1.86\\ \hline
\multicolumn{1}{|c|}{08.04.2013} & 0.4071 & 21.8992 &2.23\\ \hline
\multicolumn{1}{|c|}{22.04.2013} & 0.3578 & 28.6239 &2.56\\ \hline
\multicolumn{1}{|c|}{29.04.2013} & 0.4305 & 19.1246&2.06 \\ \hline
\end{tabular}
\end{table}

\begin{figure}
\centering
  \includegraphics[width=0.5\textwidth]{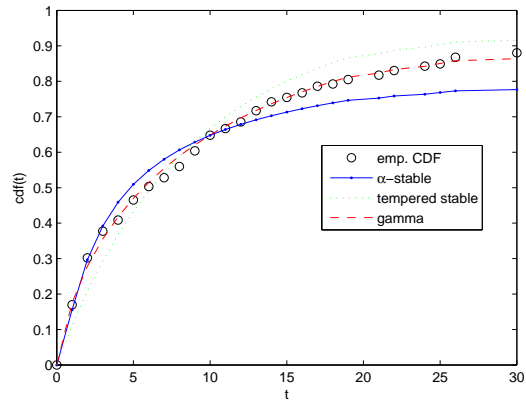}
\caption{Empirical CDF on 25th 
Feb and rescaled modified CDF of fitted distributions: $\alpha$-stable, tempered stable and gamma.}
\label{fig:cdfFittedToData}
\end{figure}

The probability of upward and downward movements is highly dependent on several factors, in particular whether lectures are held at that time and how many students have attended the lecture. Due to variability of these factors, fitting of any trend function is problematic and the results might be unreliable. Taking it into account, we have decided to model the probability based on the hour of the day, the procedure has been described in subsection \ref{sec:jumpsEstimation}. The result of the fit is presented in Fig. \ref{fig:probOfUpwardMove}.

\begin{figure}
\centering
  \includegraphics[width=0.5\textwidth]{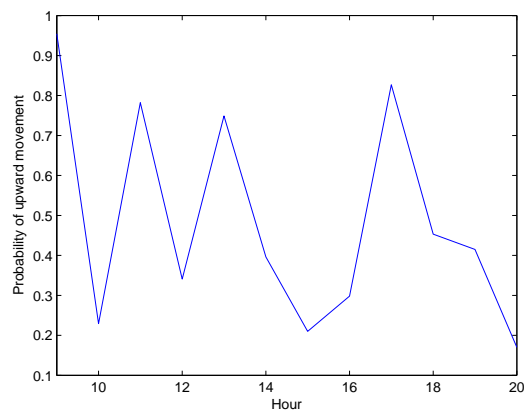}
\caption{Probability of upward movement during lecture hours.}
\label{fig:probOfUpwardMove}
\end{figure}

Taking into account the above mentioned facts we have used the estimated waiting times and probability of upward and downward movement to simulate the process. In the period marked as stale - from 09:00 till 21:00 - the process stays at the same level 400 ppm. During the lecture hours the process follows CTRW with waiting times from gamma distribution with shape equal to 0.4188 and scale equal to 21.5421. The probability of the jumps depends on the hour and has been shown in Fig. \ref{fig:probOfUpwardMove}. In addition we have applied constraints such that after 21:00 the process tends to the stale value - once it is reached it stays on this level. We have simulated 10000 trajectories and matched it with our data samples: exemplary comparison of sample from 25th March with quantile lines on levels 0.1, 0.5 and 0.9 is shown in Fig. \ref{fig:singleTraj2}. Overall we conclude that the simulated process is fitted relatively good to the data and could be used as a predictor of future values.

\begin{figure}
\centering
  \includegraphics[width=0.5\textwidth]{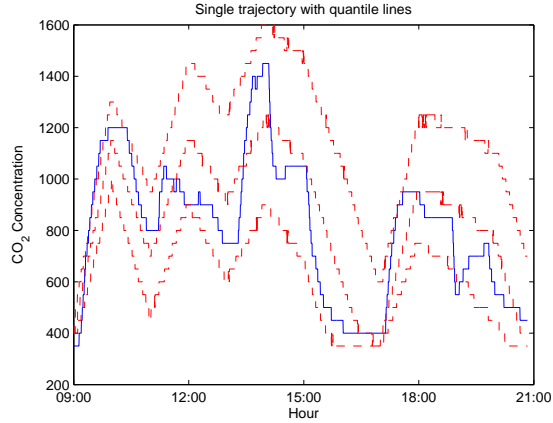}
\caption{$CO_2$ concentration on 25th March (Monday) during lecture hours together with simulated quantile lines on levels 0.1, 0.5 and 0.9.}
\label{fig:singleTraj2}
\end{figure}

{However, major practical implications of the proposed approach may be demonstrated in the domain of continuous measurements. It consists in  the precise estimation of waiting times. Such estimation allows to determine the distribution of time until the process changes significantly. From the measurement point of view it is the guideline how frequently the state of the measured parameter should be checked without losing relevant information about its temporal variation. More specific, the distribution mean may be considered as the indication of optimum sampling frequency.}
 
{In Table \ref{MondaysResults} we present the parameters of best waiting times distributions fitted to the CO$_2$ concentration data collected during 8 consecutive Mondays in lecture room. As shown, the mean values of the distribution were from 1 min 45 s to 2 min 30 s. Taking example of the day with a  minimum mean, we see that most frequently, meaningful changes of $CO_2$ concentration occurred every  1 min 45 s. Consequently, if the measurement of the parameter is performed less frequently, significant information may be lost. The inverse of the quoted time interval could be considered as the suggested sampling frequency. Smaller frequency would not be recommended. Considering that the data analyzed in this work was collected every 15 s, the obtained result of 1 min 45 s indicates the possibility of saving a considerable data storage space.  By adjusting sampling frequency, the amount of data collected is 7 times smaller. Hence, the duration of continuous measurements between consequent data download may be extended 7 times. This result is of high practical value. However, we must emphasize that it does not have a status of general recommendation. It is valid for the CO$_2$ concentration measurements performed with the defined accuracy of 50 ppm in the particular indoor environment. More studies would have to be performed in order to  obtain more generic results applicable in measurement practice in different circumstances. Because the methodology can be applied to other parameters e.g. temperature or relative humidity the approach presented in this paper bares considerable consequences for planning indoor air quality monitoring networks, in particular using of portable devices.}


\section{Conclusions}
\label{sec:conclusions}
In this paper we have considered a stochastic system that allows for modeling the carbon dioxide concentration, one of the main parameters of indoor air quality. The applied model is based on the continuous time random walk, the process which exhibits anomalous diffusion behavior. In the CTRW model here we have mainly concentrated on the waiting times  and proposed a new techniques for estimation and statistical investigation of their distribution. The methods are based on the modified cumulative distribution function,  
an extended version of classical cumulative distribution function, which is more appropriate to distribution description in case the real data are rounded. This appears also in case of waiting times visible in CTRW trajectories. By using simulated data we have proved the efficiency of proposed techniques. As examples we have considered three waiting times distributions, commonly used in practice, namely $\alpha-$stable, tempered stable and gamma. Moreover we have supported theoretical and simulated results with the real data analysis.\\
The approach was applied to analyze time series of $CO_2$ concentration monitoring data recorded in indoor environment.  It was demonstrated that the method leads to valuable conclusions concerning sampling frequency.  This is one of crucial factors which has to be chosen while planning continuous measurements although most frequently, its best value is not known.  As the presented approach is general enough to be applied to virtually any measured parameter, the method may be useful in determining sampling frequency in continuous measurements in general.

\section*{Acknowledgements}
This  paper  is  a  result  of  the  project:  "The  variability  of  physical  and  chemical  parameters  in  time  as  the  source  of  comprehensive  information  about  indoor  air  quality".  The  research of M.M, A.S. and A.W. is co-financed by the National Science Center, Poland, under the contract
No. UMO-2012/07/B/ST8/03031. \\
The research of J.G. was partially supported by NCN Maestro Grant No. 2012/06/A/ST1/00258.\\



\begin{thebibliography}{}
	
	
	\bibitem{bib:Montrol1965}
	Montroll, E.W., Weiss, G.H.: Random walks on lattices. II. J. Math. Phys. \textbf{6}, 167–181 (1965)
	
	\bibitem{metzler2000}
	Metzler, R., Klafter, J.: The random walk's guide to anomalous diffusion: a fractional dynamics approach. Phys. Rep. \textbf{339}, 1–77 (2000)
	
	\bibitem{sub1}
	Magdziarz, M.: Langevin picture of subdiffusion with infinitely divisible waiting times. J. Stat. Phys. \textbf{135}, 763–772 (2009)
	
	\bibitem{sub2} 
	Magdziarz, M.: Stochastic representation of subdiffusion processes with time-dependent drift. Stoch. Process. Appl. \textbf{119}, 3416–3434 (2009)
	
	\bibitem{sub3}
	Magdziarz, M.,Weron, A., Weron, K.: Factional Fokker-Planck dynamics: Stochastic representation and computer simulation. Physics Rev. E \textbf{75}, 016708 (2007)
	
	\bibitem{sub4}Magdziarz, M.: Path properties of subdiffusion - a martingale approach. Stochastic Models \textbf{26}, 256-271 (2010)
	
	\bibitem{sub5}
	Piryatinska, A., Saichev, A.I., Woyczynski, W.A.: Models of anomalous diffusion: the subdiffusive case. Physica A \textbf{349}, 375–420 (2005)
	
	\bibitem{a2}Scher, H., Montroll, E.: Anomalous transit-time dispersion in amorphous solids. Phys. Rev. B \textbf{12}, 2455-2477 (1975)
	
	\bibitem{a3}Scher, H., Lax, M.: Stochastic Transport in a Disordered Solid. I. Theory. Phys. Rev. B \textbf{7}, 4491-4502  (1973)
	
	\bibitem{a4}Pfister, G., Scher, H.: Dispersive (non-Gaussian) transient transport in disordered solids. Adv. Phys. \textbf{27}, 747-798 (1978)
	
	\bibitem{a5} Ott, A., Bouchaud J.-P., Langevin. D., Urbach W.:  Anomalous diffusion in "living polymers": A genuine Levy flight? Phys. Rev. Lett. \textbf{65}, 2201-2204 (1990)
	
	\bibitem{a6} Caspi, A., Granek, R., Elbaum, M.: Enhanced diffusion in active intracellular transport. Phys. Rev. Lett. \textbf{85}, 5655-5658 (2000)
	
	\bibitem{a7} Golding, I., Cox, E.C.: Physical nature of bacterial cytoplasm. Phys. Rev. Lett. \textbf{96}, 098102 (2006), 
	
	\bibitem{bib:Orzel2011}
	Orze{\l}, S., Wy{\l}oma{\'n}ska, A.: Calibration of the subdiffusive arithmetic Brownian motion with tempered stable waiting-times. J. Stat. Phys. \textbf{143}(3), 447-454 (2011)
	
	\bibitem{Janczura_orzel}
	Janczura, J., Orze{\l}, S., Wy{\l}oma{\'n}ska, A.: Subordinated $\alpha$-stable Ornstein–Uhlenbeck process as a tool for financial data description, Physica A \textbf{390}, 4379-4387  (2011)
	
	\bibitem{bib:Maciejewska2012}
	Maciejewska, M., Szczurek, A., Sikora, G., Wy{\l}oma{\'n}ska, A.: Diffusive and subdiffusive dynamics of indoor microclimate. A time series modeling. Phys. Rev. E \textbf{86}, 031128 (2012)
	
	\bibitem{bib:Janczura2013}
	Janczura, J., Maciejewska, M., Szczurek, A., Wy{\l}oma{\'n}ska, A.: Stochastic modeling of indoor air temperature. J. Stat. Phys. \textbf{152}(5), 979-994 (2013)
	
	\bibitem{bib:Maciejewska2015}
	Szczurek, A., Maciejewska, M., Teuerle,  M., Wy{\l}oma{\'n}ska, A.: Method to characterize collective impact of factors on indoor air. Physica A \textbf{420}, 190-199 (2015)
	
	\bibitem{bib:MagdziarzGajda2010}
	Gajda, J., Magdziarz, M.: Fractional Fokker-Planck equation with tempered alpha-stable waiting times. Langevin picture and computer simulation. Phys. Rev. E \textbf{82} 011117 (2010)
	
	\bibitem{inne2}
	Stanislavsky, A., Weron, K., Weron, A.: Diffusion and relaxation controlled by tempered $\alpha$-stable processes. Phys. Rev. E \textbf{78}, 051106 (2008)
	
	\bibitem{inne3}
	Stanislavsky, A., Weron, K., Weron, A.: Anomalous diffusion with transient subordinators: A link to compound relaxation laws. Journal of Chemical Physics \textbf{140}, 054113 (2014) 
	
	\bibitem{bib:Janczura2011}
	Janczura, J., Wy{\l}oma{\'n}ska, A.: Anomalous diffusion models: different types of subordinator distribution. Acta Phys. Pol. B \textbf{43}(5), 1001-1016 (2012)
	
	
	
	\bibitem{bib:monitoring1}
	Ramirez-Nino, J., Pascacio, A., Carrillo, J.,  de la Torre, O.: Monitoring network for online diagnosis of power generators. Measurement \textbf{42}, 1203-1213 (2009)
	
	\bibitem{bib:monitoring2}
	Ferreira, L.F., Antunes, P. , Domingues, F.,  Silva, P.A.,  André, P.S.: Monitoring of sea bed level changes in nearshore regions using fiber optic sensors. Measurement \textbf{45} 1527-1533 (2012).
	
	\bibitem{bib:performance}
	 Zhao, L.,  Xie,Y., Wang, J.,  Xu, X.,: A performance assessment and adjustment program for air quality monitoring networks in Shanghai, Atmos.  Environ. \textbf{122},  382-392 (2015).
	
	\bibitem{bib:wireless}
	Martinez-Garrido, M.I., Fort,  R.: Experimental assessment of a wireless communications platform for the  built  and  natural heritage, Measurement, In Press.
	
	\bibitem{bib:Szczurek2015}
	Maciejewska, M., Szczurek, A.: Indoor air quality network design based on uncertainty and mutual information, SENSORNETS (2014)
	
	\bibitem{bib:sampling}
	Govindarajulu, Z., Elements of sampling theory and methods, Prentice Hall Inc., 1999.
	
	\bibitem{bib:Wylomanska2012}
	Wy{\l}oma{\'n}ska, A.: Arithmetic Brownian motion subordinated by tempered stable and inverse tempered stable processes. Physica A \textbf{391}(22), 5685-5696 (2012)
		
	\bibitem{bib:Klafter2011}
	Klafter, J.,  Sokolov, IM.: First steps in random walks. From tools to applications. UK:Oxford University Press, Oxford, (2011)
	
	\bibitem{bib:Magdziarz2009}
	Magdziarz, M.: Langevin picture of subdiffusion with infinitely divisible waiting times. J. Stat. Phys. \textbf{135}, 763-772 (2009)
	
	\bibitem{bib:Magdziarz22009}
	Magdziarz, M.: Stochastic representation of subdiffusion processes with time-dependent drift. Stoch. Process. Appl. \textbf{119}, 3416-3434 (2009)
	
	\bibitem{bib:MeerschaertInfinite}
	Meerschaert, M.M.,  Scheffler, H.P.: Limit theorems for continuous time random walks with infinite mean waiting times. J. Appl. Probab. \textbf{41}(3), 623-638  (2004)
	
	
	\bibitem{bib:Gajda2013}
	Gajda, J., Wy{\l}oma{\'n}ska, A.: Tempered stable Levy motion driven by stable subordinator. Physica A \textbf{392}, 3168-3176 (2013)
	
	\bibitem{bib:AgnieszkaJanusz2014}
	Gajda, J.,  Wy{\l}oma{\'n}ska,  A.: Fokker-Planck type equations associated with fractional Brownian motion controlled by infinitely divisible processes. Physica A \textbf{405}, 104-113 (2014)
	
	\bibitem{bib:JanickiWeron}
	Janicki, A., Weron, A.: A Simulation and Chaotic Behavior of $\alpha$-Stable Stochastic Processes. Dekker, New York (1994)
	
	\bibitem{bib:Meerschaert2010}
	Baeumer, B., Meerschaert, M.M.: Tempered stable Levy motion and transient super-diffusion. J. Comput. Appl. Math. \textbf{233}, 2438-2448 (2010)
	
	\bibitem{bib:Whitt}
	Whitt, W.: Stochastic-Process Limits. An Introduction to Stochastic-Process Limits and Their  Applications to Queues. Springer, New York (2002)
	
	\bibitem{bib:Zebrowski2012}
	Magdziarz, M., Szczotka, W., {\.Z}ebrowski, P: Langevin Picture of Lévy Walks and Their  Extensions. J. Stat. Phys. \textbf{147}, 74-96 (2012)
	
	
	\bibitem{bib:Wylomanska2016}
	Wy{\l}oma{\'n}ska, A., Gajda, J.: Stable continuous-time autoregressive process driven by stable subordinator. Physica A \textbf{444}, 1012-1026 (2016)
	
	\bibitem{bib:Gajda2012_2}
	Gajda, J., Wy{\l}oma{\'n}ska, A.: Geometric Brownian motion with tempered stable waiting times, J. Stat. Phys. \textbf{148}, 296-305 (2012)
	
	
	
	
	
	
	\bibitem{bib:Burnecki2012}
	Burnecki, K., Wy{\l}oma{\'n}ska, A., Beletskii, A., Gonchar, V., Chechkin, A.: Recognition of stable distribution with Levy index alpha close to 2. Phys. Rev. E \textbf{85}, 056711 (2012)
	
	
	
	\bibitem{bib:Lagaris1998}
	Lagarias, J. C., Reeds, J. A., Wright, M. H., Wright, P. E.: Convergence Properties of the Nelder-Mead Simplex Method in Low Dimensions. SIAM J. Optim. \textbf{9}(1), 112-147 (1998)
	
	\bibitem{bib:Gajda2012}
	Gajda, J., Sikora, G., Wy{\l}oma{\'n}ska, A.: Regime variance testing — a quantile approach. Acta Phys. Pol. B \textbf{44}(5), 1015-1035 (2012)
	
\end{thebibliography}


\end{document}